\documentclass{article}

\usepackage{epstopdf}

\usepackage{amsmath}
\usepackage{amssymb}
\usepackage{amsfonts}
\usepackage{authblk}

\usepackage{array}
\usepackage{color}

\usepackage{multirow}
\usepackage{rotating}

\newcommand{\cA}{{\cal A}}
\newcommand{\cB}{{\cal B}}
\newcommand{\cC}{{\cal C}}

\newcommand{\cL}{{\cal L}}
\newcommand{\cV}{{\cal V}}

\newcommand*{\qed}{\hfill\ensuremath{\square}}%

\newtheorem{theorem}{Theorem}[section]
\newtheorem{lemma}{Lemma}[section]
\newtheorem{corollary}{Corollary}[section]

\newtheorem{proposition}{Proposition}[section]

\newenvironment{proof}{\noindent{\bf Proof:}}{\qed}

\newcommand{\algo}[1]{
\medskip
\noindent \textbf{Algorithm {\tt #1}}\\
\nopagebreak}

\newcommand{\procend}{\hfill $\diamond$\medskip}

\begin{document}

\title{
Topology recognition with advice\thanks{This research was done during the visit of Andrzej Pelc at  Sapienza, University of Rome,  partially supported by a visiting fellowship from this university. A preliminary version of this paper appeared in the
Proceedings of the 27th International Symposium on Distributed Computing (DISC 2013).}
}

\author[1]{Emanuele G. Fusco}
\author[2]{Andrzej Pelc\footnote{Partially supported by NSERC discovery grant 8136--2013 
and by the Research Chair in Distributed Computing at the
Universit\'e du Qu\'ebec en Outaouais.}}
\author[3]{Rossella Petreschi}
\affil[1]{Department of Computer, Control, and Management Engineering  Antonio Ruberti,
Sapienza, University of Rome,
00185 Rome, Italy.
{\tt fusco@diag.uniroma1.it}.}
\affil[2]{
 D\'epartement d'informatique, Universit\'e du Qu\'ebec en Outaouais, Gatineau,
Qu\'ebec J8X 3X7, Canada. {\tt pelc@uqo.ca}. }
\affil[3]{Computer Science Department,
Sapienza, University of Rome,
00198 Rome, Italy.
{\tt petreschi@di.uniroma1.it}.}
\date{}
\maketitle

\begin{abstract}

In topology recognition,  each node of an anonymous network has to deterministically produce an isomorphic copy of
the underlying graph, with all ports correctly marked. This task is usually unfeasible without any a priori information.
Such information can be provided to nodes as advice. An oracle knowing the network can give a (possibly different)
string of bits to each node, and all nodes must reconstruct the network using this advice, after a given number of rounds
of communication. During each round each node can exchange arbitrary messages with all its neighbors and perform
arbitrary local computations. The time of completing topology recognition is the number of rounds it takes, and the size
of advice is the maximum length of a string given to nodes.

We investigate tradeoffs between the time in which topology recognition is accomplished and the minimum
size of advice that has to be given  to nodes. We provide upper and lower bounds on the minimum size of advice
that is sufficient to perform topology recognition in a given time, in the class of all graphs of size $n$ and diameter
$D\le \alpha n$, for any constant $\alpha< 1$. In most cases, our bounds are asymptotically tight.\end{abstract}

\begin{center}
\end{center}

\section{Introduction}

\subsection{The model and the problem}

Learning an unknown network by its nodes is one of the fundamental distributed tasks in networks.
Once nodes acquire a faithful labeled map of the network, any other 
distributed task, such as leader election \cite{HS,P}, minimum weight spanning tree construction \cite{A}, 
renaming \cite{ABDKPR}, etc. can be performed by nodes using
only local computations. Knowing topology also simplifies the task of routing. 
More generally, constructing a labeled map converts all distributed problems to centralized ones, 
in the sense that nodes can solve them simulating a central monitor.

If nodes are a priori equipped with unique identifiers, they can deterministically 
construct a labeled map of the network, by exchanging messages, without any additional  information about the network. However,
even if nodes have unique identities, 
 relying on them 
for the task of learning the network is not always possible. Indeed, nodes may be reluctant to reveal their identities for security or privacy reasons. 
Hence it is
important to design algorithms reconstructing the topology of the network without assuming any node labels, i.e., for anonymous networks. In this paper we are interested in deterministic solutions.

 Ports at each node of degree $d$ are 
arbitrarily numbered $0, \dots,d-1$, and there is no assumed coherence between port numbers at different nodes. A node is aware of its degree, and it knows on which port it sends or receives a message. 
The goal is, for each node, to get an isomorphic copy of the graph underlying the
network, with all port numbers correctly  marked. (Port numbers are very useful, e.g., to perform the routing task.)
There are two variants of this task:
a weaker version, that we call
{\em anonymous topology recognition}, in which the nodes of the reconstructed map are unlabeled,
and a stronger version, that we call {\em labeled topology recognition}, in which all nodes construct 
a map of the network assigning distinct labels to all nodes in the same way, and know their position in this map.
Even anonymous topology recognition is not always feasible without any a priori information given to nodes, as witnessed, e.g.,  by the class
of oriented rings in which ports at each node are numbered 0,1 in clockwise order. No amount of information exchange can help nodes to recognize the size of the oriented ring and hence to reconstruct
correctly its topology. Thus, in order to accomplish (even anonymous) topology recognition for arbitrary networks, some information must be provided to nodes. This can be done in the form of {\em advice}. An oracle knowing the network gives a (possibly different) string of bits to each node. Then nodes execute
a deterministic distributed algorithm that does not assume knowledge of the network, but uses message exchange and the advice provided
by the oracle to nodes, in order to reconstruct the topology of the network by each of its nodes. 

In this paper we study tradeoffs between the {\em size of advice} provided to nodes and the {\em time} of topology recognition. The size of advice is defined as the length of the longest string of bits
given by the oracle to nodes. For communication,
we use the extensively studied ${\cal LOCAL}$ model  \cite{Pe}. 
In this model, communication proceeds in synchronous rounds
and all nodes start simultaneously. In each round each node
can exchange arbitrary messages with all its neighbors and perform arbitrary local computations. 
The time of completing a task is the number of rounds it takes. 
The central question of the paper is:

\begin{quotation}
What is the minimum size of advice that enables (anonymous or labeled) topology recognition in a given time $T$,
in the class of $n$-node networks of diameter $D$?
\end{quotation}

It should be stressed that the reason why we adopt the ${\cal LOCAL}$ model is to focus on how deeply nodes should probe the network in order to discover the topology. Since 
nodes are anonymous, this depth of probing may often go beyond the diameter of the network, enabling the nodes to ``see'' other nodes several times, along different paths, and deduce some information
from these deeper views. At the same time, as it is always done when using the ${\cal LOCAL}$ model  \cite{Pe}, the size of messages  is ignored.  

\begin{table}
\begin{small}
\begin{center}
   \begin{tabular}{ccc}
 \cline{1-3} 
 
 \multicolumn{2}{|c}{ \multirow{2}{*}{{\sc Time}}} &\multicolumn{1}{|c|}{\multirow{2}{*} {{\sc Optimal Size of Advice} }} \cr 
\multicolumn{2}{|c}{ } &\multicolumn{1}{|c|}{ }\cr  \cline{1-3} 
\multicolumn{2}{|c}{ \multirow{3}{*}{{$2D+1$}}} &\multicolumn{1}{|c|}{\multirow{2}{*} {{$1$} }} \cr 
\multicolumn{2}{|c}{ \multirow{1}{*}{{}}} &\multicolumn{1}{|c|}{\multirow{2}{*} {{[Proposition \ref{ubTime2D+1}]} }} \cr 
\multicolumn{2}{|c}{ } &\multicolumn{1}{|c|}{ }\cr  \cline{1-3} 
  \multicolumn{2}{|c}{ \multirow{2}{*}{{ $D+k$, }}} &\multicolumn{1}{|c|}{\multirow{2}{*} {{Upper B.: $O(1+(\log n)/ k)$, Lower B.: 1} }} \cr
 \multicolumn{1}{|c}{ \multirow{2}{*}{{ for $D/2 <  k\le D$ }}}& \multicolumn{1}{c}{\multirow{1}{*}{{$\star$}}} &\multicolumn{1}{|c|}{\multirow{2}{*} {{[Theorem \ref{ubTimeD+k}]} }} \cr 
\multicolumn{1}{|c}{ }& &\multicolumn{1}{|c|}{ }\cr  \cline{1-3} 
 \multicolumn{2}{|c}{ \multirow{2}{*}{{ $D+k$, }}} &\multicolumn{1}{|c|}{\multirow{2}{*} {{$\Theta(1+(\log n)/ k)$} }} \cr 
 \multicolumn{2}{|c}{ \multirow{2}{*}{{ for $0 <  k\le D/2$ }}} &\multicolumn{1}{|c|}{\multirow{2}{*} {{[Corollary \ref{corD+k}]} }} \cr 
\multicolumn{2}{|c}{ } &\multicolumn{1}{|c|}{ }\cr  \cline{1-3} 
\multicolumn{2}{|c}{ \multirow{3}{*}{{$D$}}} &\multicolumn{1}{|c|}{\multirow{2}{*} {{$\Theta(n\log n)$} }} \cr 
\multicolumn{2}{|c}{ \multirow{1}{*}{{}}} &\multicolumn{1}{|c|}{\multirow{2}{*} {{[Corollary \ref{corD}]} }} \cr 
\multicolumn{2}{|c}{ } &\multicolumn{1}{|c|}{ }\cr  \cline{1-3} 
 \multicolumn{2}{|c}{ \multirow{2}{*}{{$D-k$,}}} &\multicolumn{1}{|c|}{\multirow{2}{*} {{$\Theta((n^2\log n)/ (D-k+1))$} }} \cr 
 \multicolumn{2}{|c}{ \multirow{2}{*}{{for $0<k\le D$}}} &\multicolumn{1}{|c|}{\multirow{2}{*} {{[Corollary \ref{corD-k}]} }} \cr 
\multicolumn{2}{|c}{ } &\multicolumn{1}{|c|}{ }\cr  \cline{1-3} 
\end{tabular}
 
\caption{\label{summary} The summary of results. All our bounds are tight, except for those in the line with the $\star$ when $D\in o(\log n)$.}
\end{center}
\end{small}
\end{table}

\subsection{Our results}
We provide upper and lower bounds on the minimum size of advice sufficient to perform topology recognition in a given time,
in the class $\cC(n,D)$ of all graphs of size $n$ and diameter $D\le \alpha n$, for any constant $\alpha< 1$.
All our upper bounds are valid even for the harder task of labeled topology recognition, while our lower bounds also apply to the easier task of anonymous topology recognition. Hence we will only use the term {\em topology recognition} for all our results.
We prove upper bounds $f(n,D,T)$ on the minimum size of advice sufficient to perform topology recognition in a given time $T$, for the class $\cC(n,D)$, by providing an assignment of advice of size $f(n,D,T)$ and an algorithm, using this advice, that accomplishes this task, within time $T$, for any network in $\cC(n,D)$.
We prove lower bounds on the minimum size of advice, sufficient for a given time $T$, by constructing graphs in $\cC(n,D)$ for which topology recognition within this time is impossible with advice of a smaller size. (Notice that, while our upper bounds would hold even if $D \approx n$, our lower bound constructions rely on the constraint $D\le \alpha n$.)

The meaningful span of possible times for topology recognition is  between 0 and 
$2D+1$. Indeed, while  advice of size $O(n^2\log n)$ permits topology recognition in time 0 (i.e., without communication), we show that topology recognition in time $2D+1$  can be done with advice of size 1, which is optimal. 

For most values of the allotted time, our bounds are asymptotically tight.
This should be compared to many results from the literature on the advice paradigm (see, e.g.,~\cite{CFIKP,FGIP,FIP1,FKL,SN}), which often either consider the size of advice  needed for feasibility of a given task, or only give isolated points in the curve 
of tradeoffs between resources (such as time) and the size of advice.

We show that, if the allotted time is $D-k$, where $0<k\le D$,  then the optimal size of advice is  $\Theta((n^2 \log n)/(D-k+1))$.
If the allotted time is $D$, then this optimal size is $\Theta(n \log n)$.
If the allotted time is $D+k$, where $0<k\le D/2$,  then the optimal size of advice is $\Theta(1+(\log n) / k)$.
The only remaining gap between our bounds is for time $D+k$,  where $D/2<k\le D$.
In this time interval our upper bound remains  $O(1+(\log n) / k)$, while the lower bound (that holds for any time) is 1. This leaves a gap  if $D\in o(\log n)$.
See Table~\ref{summary} for a summary of our results.

Our results show how sensitive is the minimum
size of advice to the time allowed for topology recognition: allowing just one round more, from $D$
to $D+1$, decreases exponentially the advice needed to accomplish this task.
Our tight bounds on the minimum size of advice also show a somewhat surprising fact that the amount of information that nodes need to reconstruct a labeled map of the network, in a given time, and that needed to reconstruct an anonymous map of the network in this time, are asymptotically the same in most cases.

\subsection{Related work}
{Many papers \cite{AKM01,CFP,CFIKP,DP,EFKR,FGIP,FIP1,FIP2,FKL,FP,GPPR02,IKP,KKKP02,KKP05,SN,TZ05} considered the problem of increasing the efficiency of network tasks by providing nodes with some information of arbitrary kind.}
This approach was referred to as
algorithms using {\em informative labeling schemes}, or equivalently, algorithms with {\em advice}.  
Advice is given either to nodes of the network or to mobile agents performing some network task.
Several authors studied the minimum size of advice required to solve the
respective network problem in an efficient way. Thus the framework of advice permits to quantify the amount of information
needed for an efficient solution of a given network problem, regardless of the type of information that is provided.

In \cite{CFIKP} the authors investigated the minimum size of advice that has to be given to nodes
to permit graph exploration by a robot. 
 In \cite{KKP05}, given a distributed representation of a solution for a problem,
the authors investigated the number of bits of communication needed to verify the legality of the represented solution.
In \cite{FIP1} the authors compared the minimum size of advice required to
solve two information dissemination problems using a linear number of messages. In \cite{FIP2} the authors
established the size of advice needed to break competitive ratio 2 of an exploration algorithm in trees.
In \cite{FKL} it was shown that advice of constant size permits to carry on the distributed construction of a minimum
spanning tree in logarithmic time. 
In \cite{EFKR} the advice paradigm was used for online problems.
In \cite{FGIP} the authors established lower bounds on the size of advice 
needed to beat time $\Theta(\log^*n)$
for 3-coloring of a cycle and to achieve time $\Theta(\log^*n)$ for 3-coloring of unoriented trees.  
In the case of \cite{SN} the issue was not efficiency but feasibility: it
was shown that $\Theta(n\log n)$ is the minimum size of advice
required to perform monotone connected graph clearing.
In \cite{IKP} the authors studied radio networks for
which it is possible to perform centralized broadcasting in constant time. They proved that
$O(n)$ bits of advice allow to obtain constant time in such networks, while
$o(n)$ bits are not enough. 

Distributed computation on anonymous networks has been investigated by many authors, e.g.,
\cite{An,ASW,BV,EPSW,ESW,FP2, KKV,Pe,YK3}
for problems ranging from leader election to computing boolean functions
and communication in wireless networks. 
In \cite{EPSW} the authors compared randomized and deterministic algorithms to solve distributed problems in anonymous networks.
In \cite{ESW} a hierarchy of distributed problems in anonymous networks was studied.
Feasibility of topology recognition for anonymous graphs with adversarial port labelings was studied in~\cite{YK3}.
The problem of efficiency of map construction by a mobile agent, equipped with a token, exploring an anonymous graph  has
been studied in \cite{CDK}. In \cite{DP} the authors investigated the minimum size of advice
that has to be given to a mobile agent, in order to enable it to reconstruct  the topology of an anonymous network or to construct its spanning tree.
Notice that the mobile agent scenario makes the problem of map construction much different from our setting. Since all the advice is given to a single agent,
breaking symmetry may be very hard. Even anonymous map construction often requires providing a large amount of information to the agent, regardless of the exploration time. 
To the best of our knowledge, tradeoffs between time and the size of advice for topology recognition have never been studied before.

\section{Preliminaries}
Unless otherwise stated, we use the word {\em graph} to mean a simple undirected connected graph without node labels, and with ports at each node of degree $d$ labeled $\{0,\ldots,d-1\}$. 
Two graphs $G=(V,E)$ and $G'=(V',E')$ are {\em isomorphic}, if and only if, there exists a bijection $f: V \longrightarrow V'$ such that the edge
$\{u,v\}$, with port numbers $p$ at $u$ and $q$ at $v$ is in $E$, if and only if, the edge $\{f(u),f(v)\}$ with port numbers $p$ at $f(u)$ and $q$ at $f(v)$ is in $E'$.

The size of a graph is the number of its nodes. Throughout the paper we consider a fixed positive constant $\alpha < 1$ and the class of graphs of size $n$ and diameter $D\le \alpha n$. We use $\log$ to denote the logarithm to the base $2$.
For a graph $G$, a node $u$ in $G$, and any integer $t$, we denote by $N_t(u)$ the set of nodes in $G$ at distance at most $t$ from $u$.

We will use the following notion from \cite{YK3}. 
The {\em view} from node $u$ in graph $G$ is the infinite tree $\cV(u)$  rooted at $u$ with unlabeled nodes and labeled ports, whose branches are infinite sequences of port numbers coding all infinite paths in the graph, starting from node $u$. The {\em truncated view} $\cV^{l}(u)$ is the truncation of this tree to depth $l \geq 0$. to level $l$, for each $l$.
 
Given a graph $G=(V,E)$, a function $f : V \longrightarrow \{0,1\}^{*}$ is called a {\em decoration} of $G$.
Notice that an assignment of advice to nodes of $G$ is a decoration of $G$.
For a given decoration $f$ of a graph $G$ we define the {\em decorated graph} $G_f$ as follows. Nodes of $G_f$ are ordered pairs $(v,f(v))$, for all nodes $v$ in $V$.
$G_f$ has an edge $\{(u,f(u)),(v,f(v))\}$ with port numbers $p$ at $(u,f(u))$ and $q$ at $(v,f(v))$, if and only if,
$E$ contains the edge $\{u,v\}$, with port numbers $p$ at $u$ and $q$ at $v$.

We define  the {\em decorated view} at depth $l$ of node $v$ in $G$, according to $f$, as the truncated view at depth $l$ of node $(v, f(v))$ in the decorated graph $G_f$. Nodes in the decorated view are labeled with the values assigned to the corresponding nodes of $G$ by the decorating function $f$.

The following two lemmas will be used in the proofs of our upper bounds.

\begin{lemma}\label{lemDominating}
Let $G$ be a graph and let $r$ be a positive integer. There exists a set $X$ of nodes in  $G$ satisfying the following conditions. 
\begin{itemize}
\item For any node $w$ of $G$ there exists a node $u$ in $X$ such that the distance between $w$ and $u$ is at most $r$.
\item For each pair $\{u,v\}$ of distinct nodes in $X$, the distance between $u$ and $v$ is larger than $r$.
\end{itemize}
\end{lemma}
\begin{proof}
The set proving the lemma is any independent dominating set of the $r$th power of the graph $G$.
\end{proof}

\begin{lemma}\label{lemUniqueLabels}
Let $G$ be a graph of diameter $D$ and let $A$ be an injective  decoration of $G$. Then each node $u$ in $G$ can accomplish topology recognition using its view, decorated according to $A$, at depth $D+1$, even without knowing $D$ a priori.
\end{lemma}
\begin{proof}
Let $l$ be the minimum integer such that the views at depth  $l$ and $l+1$, decorated according to $A$, contain the same set of values.
Then the view at depth $l+1$ contains all nodes and all edges of the graph $G$, and thus it is sufficient to reconstruct the graph $G$ and its decoration $A$.
The lemma follows from the fact that the view at depth $D$ contains all nodes of the graph.  
\end{proof}

The following proposition can be easily proved by induction on the round number.
Intuitively it says that, if two nodes executing the same algorithm have the same decorated views at depth $t$, then they behave identically for at least $t$ rounds.

\begin{proposition}\label{propSameHistory}
Let $G$ and $G'$ be two graphs, let $u$ be a node of $G$ and let $u'$ be a node of $G'$. Let $A$ be a decoration of $G$ and let $A'$ be a decoration of $G'$.
Let $\cA$ be any topology recognition algorithm.
Let $\sigma_t$ be the set of triples $\langle p,r,m\rangle$, where $m$ is the message received by node $u$ in round $r\le t$ through port $p$ when executing algorithm $\cA$ on the graph $G$, decorated according to $A$. Let $\sigma'_t$ be defined as $\sigma_t$, but for $u'$, $G'$, and $A'$ instead of $u$, $G$, and $A$.

If the view of $u$ at depth $t$, decorated according to $A$ is the same as the view of $u'$ at depth $t$, decorated according to $A'$, then $\sigma_t = \sigma'_t$.
\end{proposition}

We will use the above proposition to prove our lower bounds as follows.
If the size of advice is too small, then there are two non-isomorphic graphs $G$ and $G'$ resulting, for some node $u$ in $G$ and some node $u'$ in $G'$, in the same decorated view at the depth equal to the time available to perform topology recognition.
Hence either $u$ or $u'$ must incorrectly reconstruct the topology (even anonymous) of $G$ or $G'$.

\section{Time  $2D+1$}

We start our analysis by constructing a topology recognition algorithm that works in time $2D+1$ and uses advice of size~1.
Since we will show that, for arbitrary $D\ge 3$, there are networks in which topology recognition without advice is impossible in any time,
this shows that the meaningful time-span to consider for topology recognition is between $0$ and $2D+1$.

\algo{TR-1}
\noindent{\bf Advice:}\\
The oracle assigns bit 1 to one node (call it $v$), and bit 0 to all others. Let $A$ be this assignment of advice.

\noindent{\bf Node protocol:}\\
In round $i$, each node $u$ sends its view at depth $i-1$, decorated according to $A$, to all its neighbors;
it receives such views from all its neighbors and constructs its view at depth $i$, decorated according to $A$.
This task continues until termination of the algorithm.

Let $t$ be the smallest round number at which node $u$ sees a node with advice $1$ in its view decorated according to $A$  (at depth $t$).
Node $u$ assigns to itself a label in round $t$ as follows.
The label $\ell(u)$ is the lexicographically smallest shortest path, defined as a sequence of consecutive port numbers  (each traversed edge corresponds to a pair of port numbers), from $u$ to any node with advice 1, in its decorated view at depth $t$.
(Notice that since there can be many shortest paths between $u$ and $v$, this node can appear many times in the decorated view at depth $t$ of $u$.)
Let $A^*$ be the decoration corresponding to the labeling  obtained as above.
(We will show that labels in $A^*$ are unique.)

After round $t$, node $u$ starts constructing its decorated view, according to decoration $A^*$. 
In any round $t' >t$, node $u$ sends both its view, decorated according to $A$, at depth $t'$, and its view, decorated according to $A^*$, at the largest possible depth. Messages required to perform this task are piggybacked to those used for constructing views, decorated according to $A$, at increasing depths.
In each round $t'$, node $u$ checks for newly discovered values of $A^*$.
As soon as there are no new values, node $u$ reconstructs the labeled map and outputs it.
Then node $u$ computes the diameter $D$ of the resulting graph and continues to send its views, decorated according to $A$ and according to  $A^*$, at increasing depths, until round $2D+1$.  After round $2D+1$ node $u$ terminates.
\procend

\begin{proposition}\label{ubTime2D+1}
Algorithm {\tt TR-1} completes topology recognition for all graphs of size $n$ and diameter $D$ in time $2D+1$, using advice of size 1.
\end{proposition}
\begin{proof}
We first prove that the topology reconstructed at each node is correct.

Uniqueness of labels assigned according to  $A^*$  follows from the fact that  the node $v$ with advice $1$ is unique and, for each node $u$,
the lexicographically smallest shortest path between $u$ and $v$ is unique.
Since, in a graph of diameter $D$, every node is at distance at most $D$ from $v$, each node $u$ computes its label $\ell(u)$ within $D$ rounds.

Notice that some nodes will acquire their unique label earlier than others. However, since each node acquires its complete decorated view (according to $A^*$) at depth $i$ within $i$ rounds after all nodes at distance at most $i$ from it have acquired their unique label, after $2D+1$ rounds of communication all nodes have their decorated view (according to $A^*$) at depth $D+1$.
By making each node continue to participate in the construction of decorated views at increasing depths until round $2D+1$ (even after it reconstructed the topology), the algorithm guarantees that {\em all} nodes will be able to get their views, decorated according to $A^*$, at depth $D+1$. Correctness of the reconstructed topology follows from Lemma~\ref{lemUniqueLabels}. Since all nodes stop in round $2D+1$ and the advice provided to each node consists of 1 bit, the proof is complete.
\end{proof}

The following proposition, cf. \cite{YK3}, shows that advice of size 1, as used by Algorithm~{\tt TR-1}, is necessary, regardless of the allotted time.
As opposed to the $n$-node rings mentioned in the introduction as graphs that require at least one bit of advice, but whose diameter is $\lfloor n/2 \rfloor$, the class of graphs we will use to prove the proposition allows greater flexibility of the diameter.

\begin{proposition}\label{lb2D+1}
Let $D\ge 3$ and let $n\ge D+6$ be an even integer.
The size of advice needed to perform topology recognition for the class of all graphs of size $n$ and diameter $D$ is at least 1.
\end{proposition}

\begin{proof}
Consider the two graphs $G_1$ and $G_2$ in Figure~\ref{fig.yk}.a that where constructed in~\cite{YK3}. They have the following properties:
\begin{itemize}
\item they are non-isomorphic;
\item both of them have size 6 and diameter 3;
\item all black nodes in both graphs have the same view;
\item the black nodes are adjacent in both graphs.
\end{itemize}

\begin{figure}
\centering

\begin{tabular}{cc}
\cline{1-2}
\multicolumn{1}{|c|}{
\scalebox{0.24}{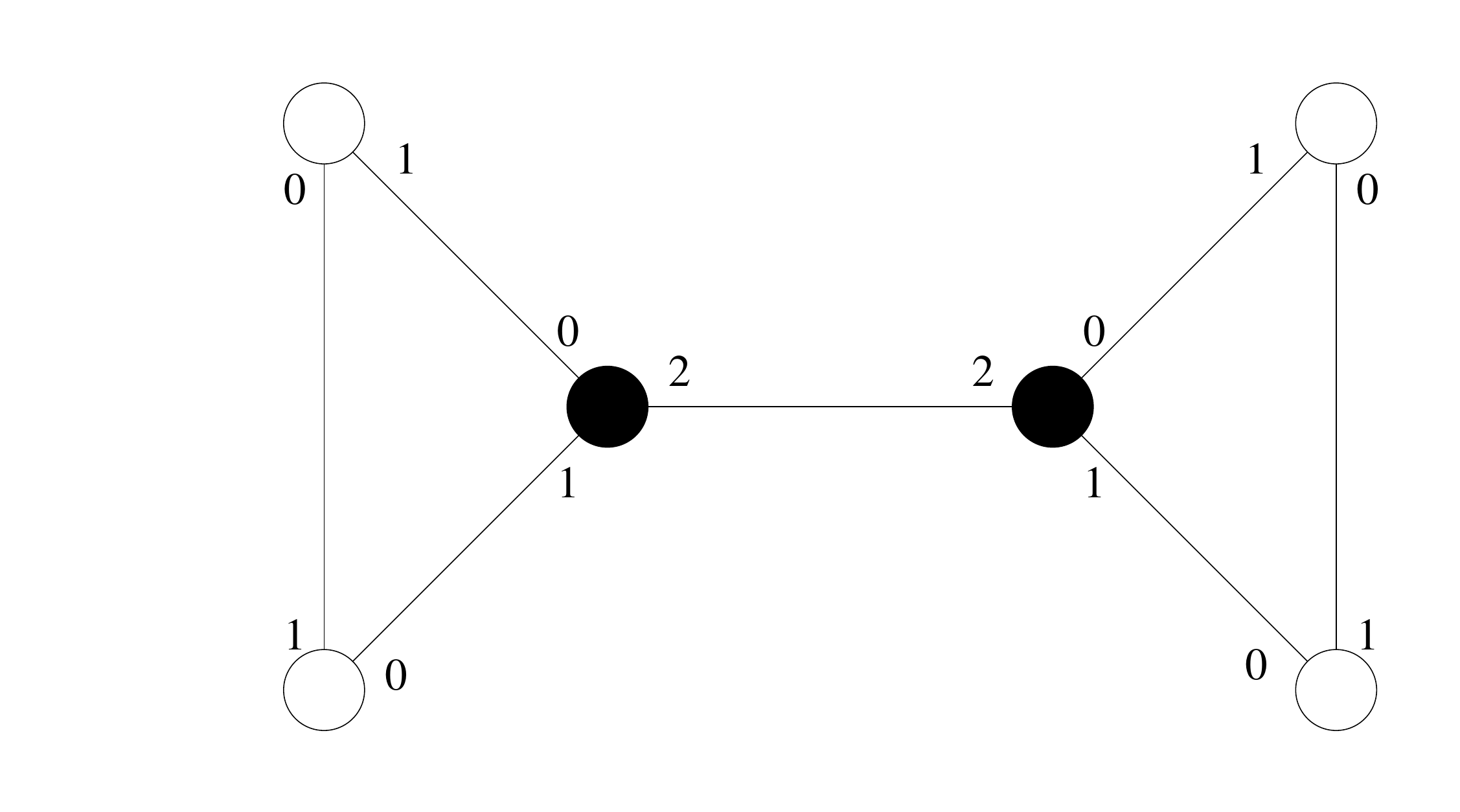}
} & \multicolumn{1}{c|}{
\scalebox{0.24}{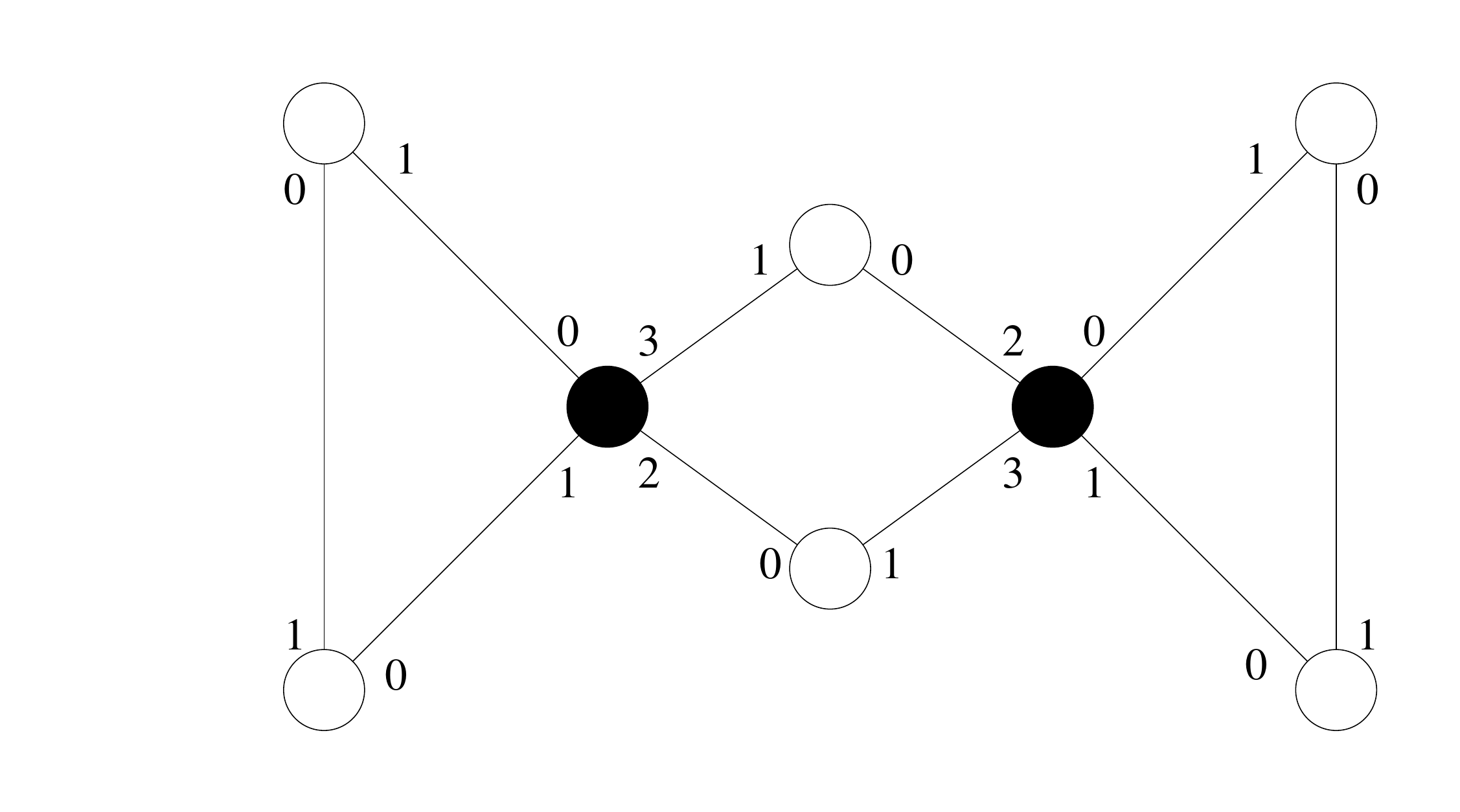}
}\\

\multicolumn{1}{|c|}{
\scalebox{0.24}{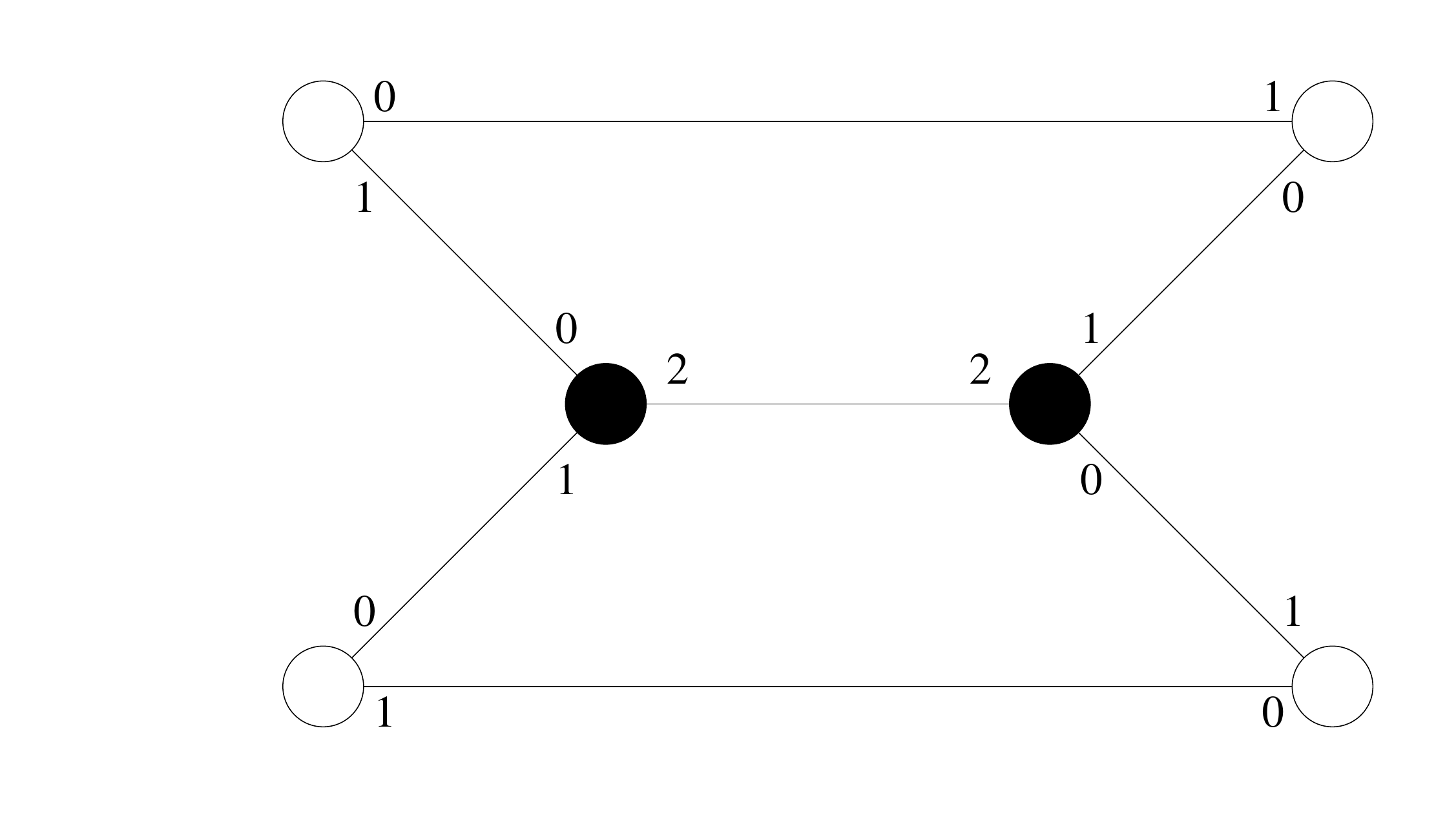}
} & \multicolumn{1}{c|}{
\scalebox{0.24}{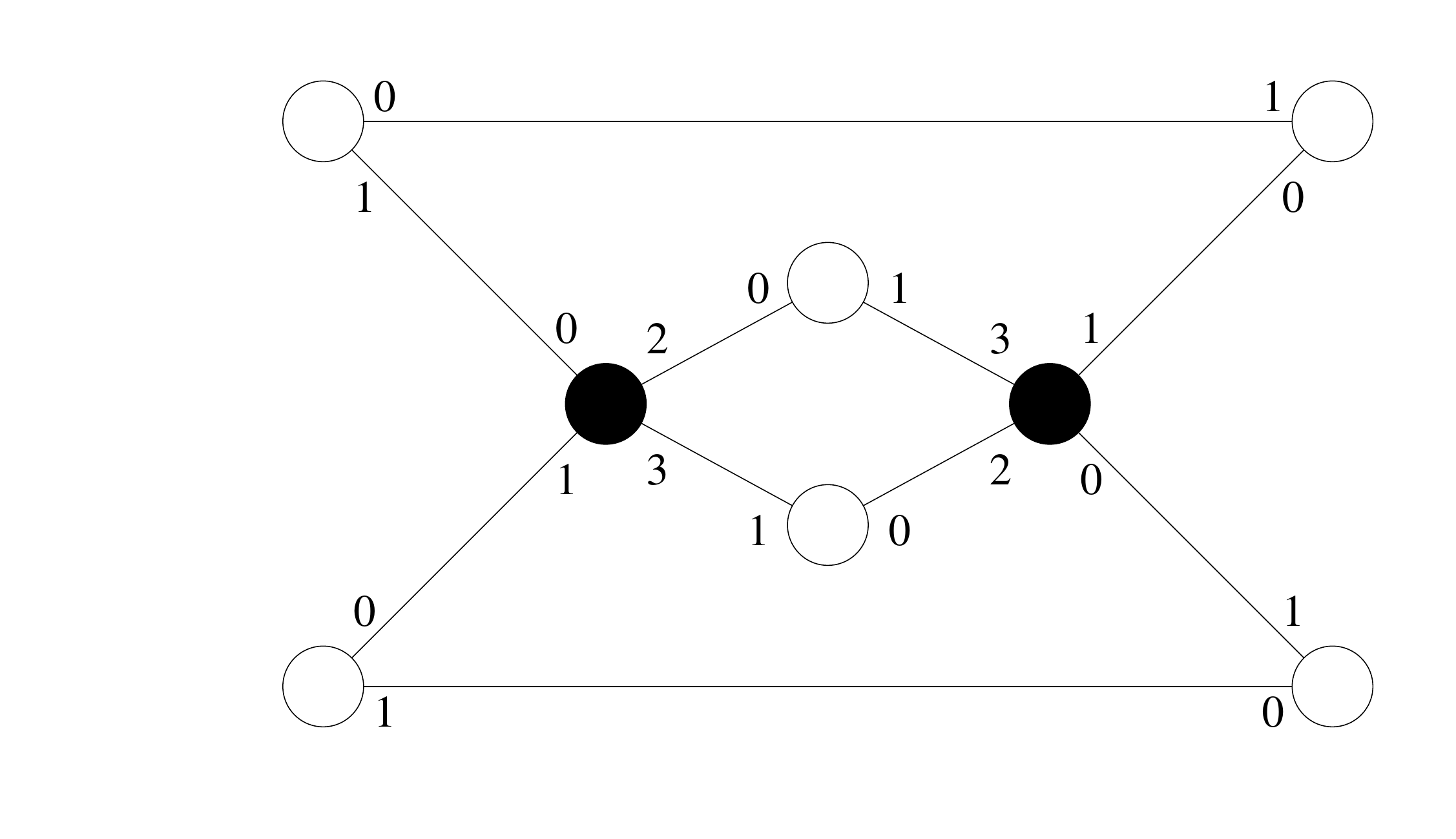}
}\\

\cline{1-2}
\multicolumn{2}{|c|}{
\scalebox{0.24}{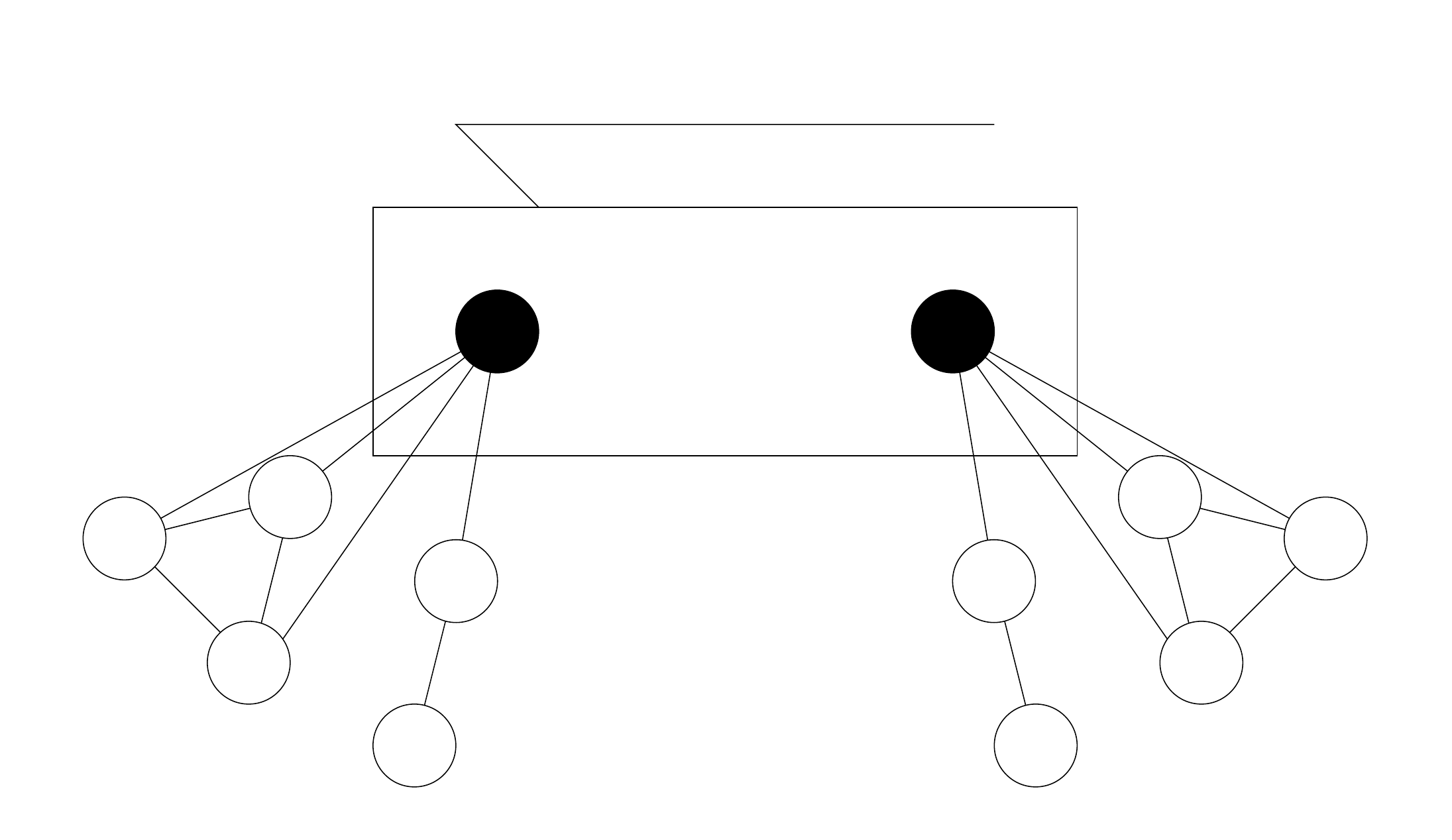}
}\\
\cline{1-2}

\end{tabular}

\caption{\label{fig.yk} a) Graphs $G_1$ and $G_2$ from~\cite{YK3}. b) Graphs $G'_1$ and $G'_2$. c) Schematic representation of graphs $H_1$ and $H_2$, in which the rectangle represents one of the graphs $G_1$, $G_2$, $G'_1$, or $G'_2$.  }
\end{figure}

The graphs $G'_1$ and $G'_2$ in Figure~\ref{fig.yk}.b are obtained from  $G_1$ and $G_2$, respectively, replacing the edge connecting the two black nodes by two paths of length 2.
These graphs have the following properties:
\begin{itemize}
\item they are non-isomorphic;
\item both of them have size 8 and diameter at most 4;
\item all black nodes in both graphs have the same view;
\item the black nodes are at distance 2 in both graphs.
\end{itemize}
For $D\ge 3$ and even $n\ge D+6$ we construct two graphs $H_1$ and $H_2$, of size $n$ and diameter $D$, as follows.
If $D$ is odd and $D=2k+1$, we obtain $H_1$ by  connecting two identical copies of a path of length $k$ and of a clique of size $(n-2k -6)/2$ to each black node of the graph $G_1$, in a symmetric fashion.
We obtain $H_2$ by doing the same for graph $G_2$.
If $D$ is even and $D=2k+2$, we obtain $H_1$ by  connecting two identical copies of a path of length $k$ and of a clique of size $(n-2k -8)/2$ to each black node of the graph $G'_1$, in a symmetric fashion.
We obtain $H_2$ by doing the same for graph $G'_2$. (See Figure~\ref{fig.yk}.c.)

Notice that graphs $H_1$ and $H_2$ have diameter $D$ and size $n$, they are non-isomorphic both for $D$ odd and for $D$ even, and all black nodes in both graphs have the same view. Hence, in view of Proposition~\ref{propSameHistory}, black nodes cannot correctly perform (even anonymous) topology recognition in $H_1$ or in $H_2$ without any advice. 
\end{proof}


\section{Time above $D$}
In this section we study the size of advice sufficient to perform topology recognition in arbitrary time larger than $D$ (but smaller than $2D+1$), i.e., large enough for allowing each node to see all nodes and edges of the graph.
We first give an algorithm using advice of size $O(1+\log(n)/k)$ that performs topology recognition in time $D+k$.

\algo{TR-2}
\noindent{\bf Advice:}\\
Let $G$ be a graph of size $n$ and diameter $D$. Let $t=\lceil k/4\rceil -1$.

If $t=0$ then the oracle gives a unique label of size $\lceil \log n\rceil$ as advice to each node.

Suppose that $t\ge 1$.
The oracle picks a set of nodes $X$ satisfying Lemma~\ref{lemDominating}, for $r=2t$.
Then it chooses a unique label $\ell(v)$ from the set $\{0, \ldots,  n -1\}$ for each node $v$ in $X$.
For any node $u \in N_{t-1}(v)$ let $\pi_v(u)$ be the lexicographically smallest shortest path (coded as a sequence of consecutive port numbers) from $u$ to $v$.
Sort the nodes $u$ in $N_{t-1}(v)$ in the increasing lexicographic order of $\pi_v(u)$.
The binary representation of $\ell(v)$ is partitioned into $|N_{t-1}(v)|$ consecutive segments, each of length at most $\lceil ( \log n )  / |N_{t-1}(v)|\rceil$.
The oracle assigns the first segment, with a trailing bit $1$,  as advice to node $v$.
For $1<i\le |N_{t-1}(v)|$, the $i$-th segment, with a trailing bit 0, is assigned as advice to the $i$-th  node of $N_{t-1}(v)$. (Notice that some nodes in $N_{t-1}(v)$ could receive only the trailing bit 0 as advice.) All other nodes get the empty string as advice.
Let $A_1$ be the above assignment of advice.

\noindent{\bf Node protocol:}\\
We first describe the protocol when $t\ge 1$.
In round $i$, each node $u$ sends its view at depth $i-1$, decorated according to $A_1$, to all its neighbors;
it receives such views from all its neighbors and constructs its view at depth $i$, decorated according to $A_1$.
This task continues until termination of the algorithm.

Each node $u$ whose advice has a trailing bit $0$ assigns to itself a temporary label $\ell'(u)$ as follows.
Let $s$ be the smallest round number at which node $u$ sees a node with advice with a trailing bit $1$ in its view decorated according to $A_1$  (at depth $s$).
The label $\ell'(u)$ is the lexicographically smallest shortest path, defined as a sequence of consecutive port numbers, from $u$ to any node with advice with a trailing bit 1 in its  view, decorated according to $A_1$, at depth $s$.

Let $u$ be a node whose advice has a trailing bit $0$. After reconstructing its label $\ell'(u)$, node $u$ sends $(\ell'(u),A_1(u))$ to the node $v\in X$ closest to it, along the lexicographically smallest shortest path that determined label $\ell'(u)$. Nodes along this path relay these messages piggybacking them to any message that they should send in a given round.

Each  node $v\in X$ (having a trailing bit 1 in its advice) computes $t$ as the first depth in which its view, decorated according to $A_1$ contains nodes without any advice.
In round $2t$ each such node reconstructs its label $\ell(v)$ from messages $(\ell'(u), A_1(u))$ it received (which it sorts in the increasing lexicographic order of $\ell'(u)$), and from $A_1(v)$.

Let $A_2$ be the decoration of $G$ where each node $v$ in the set $X$ is mapped to the binary representation of its unique label $\ell(v)$, and each node outside of $X$ is mapped to the empty string.

Nodes outside of $X$  start constructing their decorated view, according to $A_2$. This construction is put on hold by a node $v$ in $X$ until the time when it reconstructs its unique label $\ell(v)$.
Upon reconstructing its label $\ell(v)$, each node $v \in X$ starts constructing its view decorated according to $A_2$, hence allowing its neighbors to construct their view, decorated according to $A_2$, at depth 1.
This process continues for $2t$ steps, during which nodes construct and send their views at increasing depth, decorated according to $A_2$.

Each node $u$ assigns a label $\ell'' (u)$ to itself as follows.
Let $s'$ be the smallest depth at which the view of $u$, decorated according to $A_2$, contains a node $v$ with label $\ell(v)$ and let $\lambda(u,v)$ be the lexicographically smallest path connecting $u$ to such a node $v$ (coded as a sequence of consecutive port numbers). Node $u$ sets $\ell''(u)=(\lambda(u,v), \ell(v))$.


Let $A_3$ be the decoration of $G$ where each node $u$ is mapped to $\ell''(u)$.
(We will prove that $A_3$ is an injective function.)
Upon computing its value in $A_3$ each node starts constructing its decorated view, according to $A_3$.
In each round $t'$, node $u$ checks for newly discovered values of $A_3$.
As soon as there are no new values, node $u$ reconstructs the labeled map and outputs it.
Then node $u$ computes the diameter $D$ of the resulting graph and continues to send its views, decorated according to $A_1$, according to  $A_2$, and according to $A_3$, at increasing depths, until round $D + 4t+1$.  After round $D + 4t+1$ node $u$ terminates.

If $t=0$, the protocol consists only of the last step described above, with decoration $A_3$ replaced by the assignment of advice given to nodes by the oracle.
\procend

\begin{theorem}\label{ubTimeD+k}
Let $0<k\le D$.
Algorithm {\tt TR-2} completes topology recognition for all graphs of size $n$ and diameter $D$ within time $D+k$, using advice of size $O(1+ (\log n)/ k)$.
\end{theorem}
\begin{proof}
Let $G$ be a graph of size $n$ and diameter $D$. If $t=\lceil k/4\rceil -1=0$, then $k$ is constant and each node receives advice of size $\lceil \log n\rceil$. In this case the proof follows from Lemma~\ref{lemUniqueLabels}.
Hence in the sequel we can assume that $t\ge 1$.

Let $X$ be the set of nodes satisfying Lemma~\ref{lemDominating}, for $r=2t$,  selected by the oracle.
Due to the properties of the set $X$, the sets $N_{t}(v)$ and consequently the sets $N_{t-1}(v)$
 are pairwise disjoint for distinct nodes $v\in X$.
Hence, no node would receive more than one segment of a label $\ell(v)$, for some node $v\in X$.

Within round $2t-2$, each node $v\in X$ receives the complete set of segments $A_1(u)$,  for $u\in N_{t-1}(v)$, of its label $\ell(v)$.
Moreover, it can reconstruct the order of the segments according to the lexicographic order of temporary labels $\ell'(u)$ received together with the corresponding segments. 
Hence, by round $2t$, all nodes in $X$ know their unique label $\ell(v)$.

Since each node of $G$ is at distance at most $2t$ from some node in $X$ and $A_2$ is an injective function for nodes in $X$, the view decorated according to $A_2$ at depth $2t$ of each node contains some uniquely labeled node.
Consequently, $A_3(u)$ can be computed by each node $u$ within $4t$ rounds.
The decoration $A_3$ is an injective function, due to uniqueness of labels $\ell(v)$ and due to the fact that, if nodes $u\ne u'$ used the same node $v\in X$ to compute $\lambda(u,v)$ and  $\lambda(u',v)$, respectively, then $\lambda(u,v)\ne \lambda(u',v)$.
In view of  Lemma~\ref{lemUniqueLabels}, all nodes accomplish (labeled) topology reconstruction within additional $D+1$ rounds after all nodes have computed their value according to decoration $A_3$. Hence the algorithm terminates within time $D+4t+1\le D+k$.

Since the size of the set  $N_{t-1}(v)$ is at least $t$, the size of advice is at most 
$\lceil ( \log n)  / t\rceil +1 \in O(1+(\log n) / k)$, which completes the proof of the theorem.
\end{proof}

We now provide a lower bound on the minimum size of advice sufficient to perform topology recognition. This bound matches the upper bound given by Algorithm~{\tt TR-2} in the time-interval $[D+1, \ldots, 3\lfloor D/2\rfloor]$.

\begin{theorem}\label{lbD+k}
Let $2\le D\le  \alpha n$ and $0 < k \le  D/2$.
The size of advice needed to perform topology recognition in time $D+k$ in the class of graphs of size $n$ and diameter $D$
is in $\Omega((\log n) / k)$.
\end{theorem}
\begin{proof}
Our lower bound will be proved using the following classes $\cB(n,D,k)$ of graphs of size $n$ and diameter $D$, called {\it brooms}.
We define these classes for $n$ sufficiently large and for $k<\log n$. {(For $k\ge \log n$ Proposition~\ref{lb2D+1} applies.)}
Nodes in a  broom $B \in \cB(n,D,k)$ are partitioned into three sets, called the {\em bristles}, the {\em stick}, and the {\em handle}.
Let $m$ be the largest even integer such that $km + D - k < n$.

The set  bristles consists of $km$ nodes, partitioned into $m$ pairwise disjoint sets $B_1, \ldots,B_m$, each of size $k$. The stick consists of $D-k$ nodes, and the remaining $n - (km + D - k)$ nodes are in the handle.
Hence the bristles, the stick, and the handle are non-empty sets.
We now describe the set of edges in each part.

Edges of the bristles are partitioned into two sets, $E_1$ and $E_2$. 
Edges in $E_1$ connect nodes of each set $B_i$ into a path with port numbers 0 and 1 at each edge. We call  {\it head} of each set $B_i$ the endpoint of the path to which port number 1 has been assigned, and {\it tail} the other  endpoint (to which port number 0 has been assigned).
Notice that sets $B_i$ can be of size 1, in which case heads coincide with tails.
Edges in $E_2$ form a perfect matching $M$ among tails of the bristles. These edges have port number 1 at both endpoints. 

Edges of the stick form a path of length $D-k-1$ with port numbers 0 and 1 at each edge. (Notice that this path is of length 0, i.e., the stick consists of a single node, when $D=2$.) 
The handle has no edges.

The bristles, the stick, and the handle are connected as follows.
Let $u$ be the endpoint of the stick to which port number 1 has been assigned, and let $v$ be the other endpoint of the stick (to which port number 0 has been assigned). Nodes $v$ and $u$ coincide when $D=2$.
Node $v$ is connected to the head of each set $B_i$ by an edge with port numbers $i$ at $v$ and 0 at each head.
Node $u$ is connected to each node in the handle. Port numbers at $u$ corresponding to these connecting edges are numbered $\{0, 2, \dots, n - (km + D - k) \}$, if $u\neq v$, and 
$\{0, m+1, \dots, n - (k-1) m -2 \}$, if $u=v$. Nodes in the handle are of degree 1, so they have a unique port with number 0. See Figure~\ref{fig.fan} for an example of a broom in $\cB(23,6,3)$.
Notice that all brooms in $\cB(n,D,k)$ are defined over the same set of nodes and share the same edges, apart from those in sets forming perfect matchings among tails of the bristles. Moreover notice that growing the length of the bristles above $\lfloor D/2\rfloor$ would result in a graph of diameter larger than $D$ {(due to the distance between any pair of unmatched tails of the bristles)}, which explains the assumption $k\le D/2$.

\begin{figure}
\centering
\includegraphics[scale=0.45]{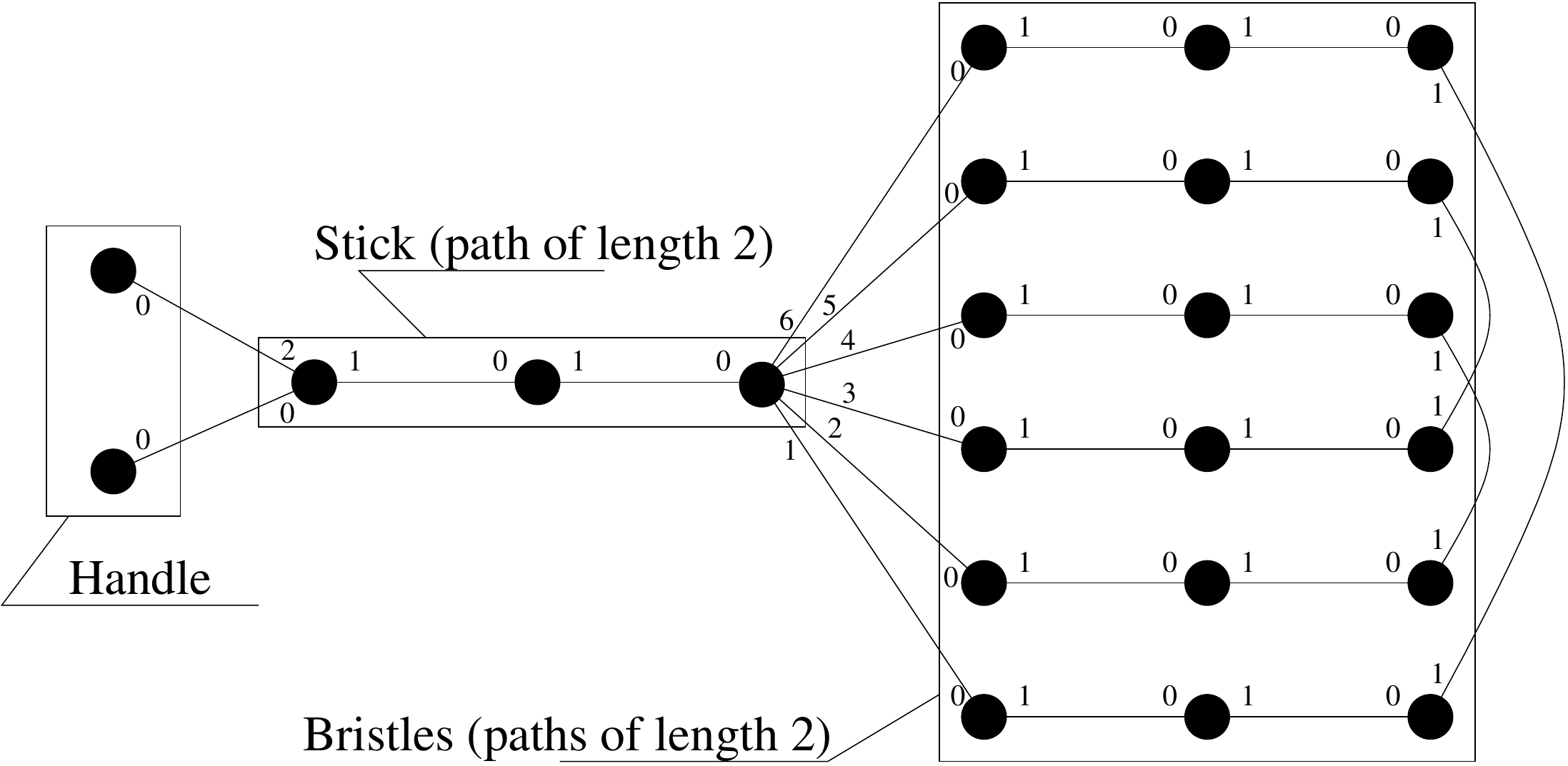}
\caption{\label{fig.fan} A broom in $\cB(23, 6,3)$.}
\end{figure}

For two brooms $B'$ and $B''$ in $\cB(n,D,k)$ we define {\em corresponding} nodes as follows. Let $h'$ and $h''$ in $B'$ and $B''$, respectively, be the nodes in the handles whose only incident edge has port number 0 at both endpoints.
Node {$u'\in B'$} corresponds to node {$u''\in B''$}, if and only if, the (unique) shortest path (defined as the sequence of  port numbers on consecutive edges) from $h'$ to $u'$ is the same as the shortest path from $h''$ to $u''$.

The idea of the proof is to show that if the size of advice is smaller than $c (\log n) / k$, for a sufficiently small constant $0<c<1$, then there exist two brooms in $\cB(n,D,k)$, whose corresponding nodes receive the same advice, for which the decorated view at depth $D+k$ of each node in the handle remains the same in both brooms. 
Since different brooms are non-isomorphic, this will imply the theorem, in view of Proposition~\ref{propSameHistory}.

Observe that for $k \in  \Omega(\log n)$, we have that $c(\log n) / k$ is constant, and  $\Omega(1)$ is a lower bound on the size of advice for topology recognition, regardless of the allowed time.
Hence we do not need to define brooms when $k \ge \log n$ to prove the theorem.

We now provide a lower bound on the size of the class $\cB(n,D,k)$. 
This size depends on the number $m$ of tails of the bristles among which perfect matchings $M$ can be defined.
For given $n$ and $k$, the size of the class $\cB(n,D,k)$ cannot increase when $D$ grows. 
Hence the class is smallest for the largest considered value of $D$, i.e., $D=\lfloor \alpha n\rfloor$.
We do the estimation for this value of $D$.

The number of perfect matchings among tails is at least $(m-1)\cdot (m-3)\cdot (m-5)\cdot  \ldots \cdot 3\cdot 1> (m/2)! $.


Suppose, from now on, that the size of advice is bounded by $c (\log n)/k$, for some constant $0<c<1$.
Then there are at most $2^{(c(\log n)/k +1) n}$ ways of assigning advice to nodes of a broom in $\cB(n,D,k)$. 
Hence there are at least $(m/2)!/2^{(c(\log n)/k +1) n}$ brooms in $\cB(n,D,k)$ for which corresponding nodes get the same advice. Fix one such assignment $A$ of advice.

We now provide an upper bound on the number of distinct decorated views, at depth $D+k$, of any node in the handle, when advice is assigned to nodes according to $A$.
Consider two brooms $B'$ and $B''$ in $\cB(n,D,k)$, decorated according to assignment $A$.
Let $M'$ and $M''$ be the perfect matchings among tails of the bristles corresponding to brooms $B'$ and $B''$, respectively.
$B'$ and $B''$ result in distinct decorated views, at depth $D+k$, of corresponding nodes in the handle, if and only if, there exist corresponding tails of the bristles $t_i'\in B'$ and $t''_i\in B''$, such that the decorated path $B'_j$, whose tail  $t'_j$ is matched to $t'_i$ in $M'$ and the decorated path $B''_h$, whose tail  $t''_h$ is  matched to $t''_i$ in $M''$, are different.
The number of distinct decorated paths $B_i$ of length $k-1$ is at most $x=2^{(c(\log n)/k +1 )k}$.
Since $m\le n/k$, it follows that there are at most $x^{n/k}= 2^{(c(\log n)/k +1 )n}$ distinct decorated views, at depth $D+k$, for any node in the handle, for assignment $A$.

We will show that the following inequality
$(m/2)! > 2^{2(c(\log n)/k +1 )n }$
which we denote by (*),
holds for $c<(1-\alpha)/128$, when $n$ is sufficiently large. 
 Indeed, for sufficiently large $n$ we have $m > n(1-\alpha)/(2k)$; in view of $k<\log n$, taking the logarithms of both sides we have 
 $$\log \left(\frac{m}{2}!\right)>\frac{m}{4} \log \frac{m}{4}>\frac{n(1-\alpha)}{8k} \log \frac{n(1-\alpha)}{8k} \ge  2\left(\frac{c\log n}{k} +1 \right)n.$$

Inequality (*) implies that 
$(m/2)! / 2^{(c(\log n)/k +1 )n} > 2^{(c(\log n)/k +1 )n}$. Hence the number of brooms from $\cB(n,D,k)$, decorated according to assignment $A$, exceeds the number of distinct decorated views, at depth $D+k$, of any node in the handle, for these brooms. It follows {from the pigeonhole principle}  that some decorated view corresponds to different brooms from $\cB(n,D,k)$. In view of Proposition~\ref{propSameHistory}, this proves that (even anonymous) topology recognition in time $D+k$, for the class of graphs of diameter $D$ and size $n$, requires advice of size at least $(1-\alpha)(\log n) /(128k) \in \Omega(\log n/k)$.
\end{proof}

Since the lower bound $\Omega(1)$ on the size of advice holds regardless of time, 
theorems~\ref{ubTimeD+k} and \ref{lbD+k} imply the following corollary.

\begin{corollary}\label{corD+k}
Let $D\le  \alpha n$ and $0 < k \le  D/2$.
The minimum size of advice sufficient to perform topology recognition in time $D+k$ in the class of graphs of size $n$ and diameter $D$ is in $\Theta(1 + (\log n / k))$. 
\end{corollary}
\section{Time $D$}

In this section we provide asymptotically tight upper and lower bounds on the minimum size of advice sufficient to perform topology recognition in time equal to the diameter $D$ of the network.
Together with the upper bound proved in Theorem~\ref{ubTimeD+k}, applied to time $D+1$, these bounds show an exponential gap in the minimum size of advice due to time difference of only one round.

\algo{TR-3}
\noindent{\bf Advice:}\\
The oracle assigns a unique label $\ell(u)$ from the set $\{0,\ldots, n-1\}$ to each node $u$.
The advice given to each node $u$ consists of the diameter $D$, the label $\ell(u)$, and the collection of all  edges incident to $u$, coded as quadruples $\left<\ell(u),p,q,\ell(v)\right>$, where $p$ is the port number at node $u$ corresponding to edge $\{u,v\}$, and $q$ is the port number at node $v$ corresponding to this edge.

\noindent{\bf Node protocol:}\\
In round $i$, each node sends to all its neighbors the collection of edges learned in all previous rounds. 
After $D$ rounds of communication each node reconstructs the topology and stops.
\procend

\begin{proposition}\label{ubTimeD}
Algorithm {\tt TR-3} completes topology recognition for all graphs of size $n$ and diameter $D$ in time $D$, using advice of size $O(n \log n)$.
\end{proposition}
\begin{proof}
Since each node is within distance $D$ from all others, each node learns the whole collection of edges of the graph within $D$ rounds.
Since these edges are coded using unique labels of adjacent nodes and contain port numbers, this is enough to reconstruct the (labeled) topology.

Since, in a graph of size $n$, the diameter  can be coded with $O(\log n)$ bits and there are $O(n)$ edges incident to a node, each edge being coded with $O(\log n)$ bits, the size of advice is $O(n \log n)$.
\end{proof}


The following lemma will be used for our lower bound.

\begin{lemma}\label{cliques}
There are at least {$((n-1)!)^n/(n!)$ non-isomorphic cliques of size $n$.}
\end{lemma}
\begin{proof}
{Assume first that the nodes are labeled $0, \ldots, n-1$. For each node $i\in [0, n-1]$ there are $(n-1)!$ ways to assign ports at it. These assignments are independent for each node, hence there are $((n-1)!)^n$ distinct assignments of port numbers.
Since from one clique without node labels it is possible to obtain at most $n!$ distinct cliques with nodes labeled $0,\ldots,n-1$, our lower bound follows.}
%
\end{proof}

We define the following classes $\cL(n,D)$ of graphs of size $n$ and diameter $D\le \alpha n$, called {\it lollipops}.
These graphs will be used to prove our lower bounds for time $D$ and below.
Nodes in a  lollipop $L \in \cL(n,D)$ are partitioned into two sets, called the {\em candy} and the {\em stick}.
The candy consists of $n-D$ nodes; for the purpose of describing our construction we will call these nodes $w_1,\ldots, w_{n-D}$. The stick consists of the remaining $D$ nodes {(see fig.~\ref{fig.lollipop} for an example of a lollipop in $\cL(6,3)$)}. 

\begin{figure}
\centering
\includegraphics[scale=0.45]{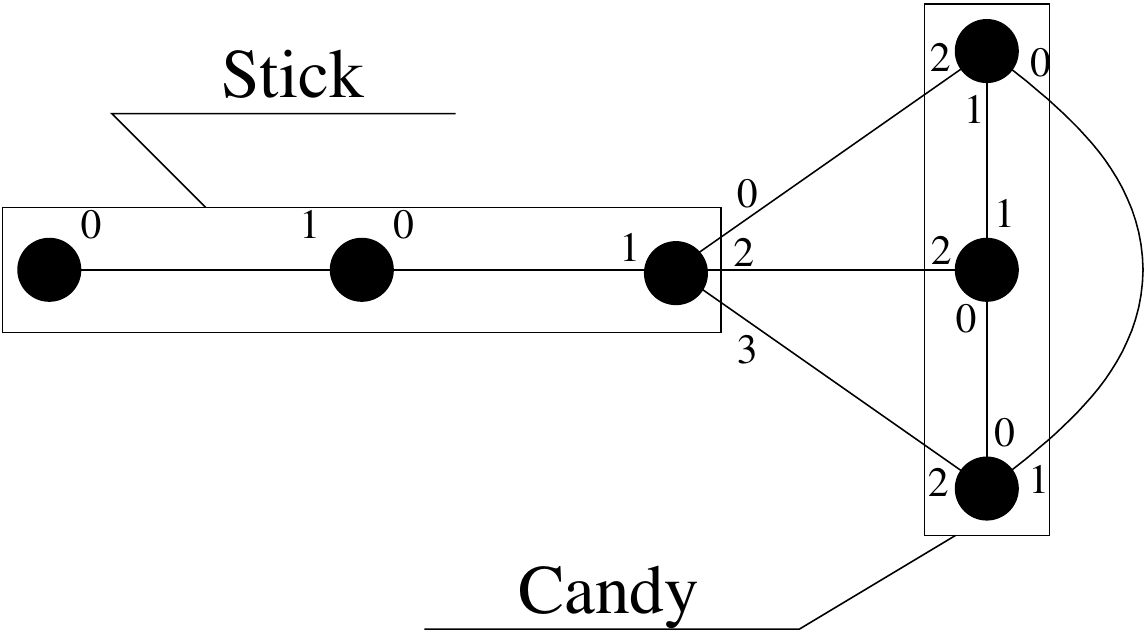}
\caption{\label{fig.lollipop} A lollipop in $\cL(6,3)$.}
\end{figure}

Nodes in the candy are connected to form a clique; port numbers for these edges are assigned arbitrarily from the set $\{0, \ldots, n-D-2\}$. 
Edges of the stick form a path of length $D-1$ with port numbers 0 and 1 at each edge. 

The stick and the candy are connected as follows.
Let $v$ be the endpoint of the stick to which port number 1 has been assigned and let $u$ be the other endpoint of the stick (to which port number 0 has been assigned). Notice that $u$ and $v$ coincide, when $D=1$.
Node $v$ is connected to all nodes in the candy. The port number, at node $v$, corresponding to edge  $\{v, w_i\}$ is $0$, if $i=1$. For $i>1$ this port number is $i$, when $u\ne v$ and $i-1$, when $u=v$. The port number, at all nodes $w_i$, corresponding to edge $\{v, w_i\}$ is $n-D-1$.

Since, for $D\le  \alpha n$, the size of the candy of a lollipop in $\cL(n,D)$ is at least $\lceil n(1-\alpha)\rceil$, Lemma~\ref{cliques} implies the following corollary:

\begin{corollary}\label{corLollipops}
The size of the class $\cL(n, D)$, for $D\le  \alpha n$ is at least {$((\lceil n(1-\alpha) \rceil -1)!)^{\lceil n(1-\alpha) \rceil} /(\lceil n(1-\alpha) \rceil !)\,$.}
\end{corollary}

\begin{theorem}\label{lbD}
Let $D\le \alpha n$.
The size of advice needed to perform topology recognition in time $D$ in the class of graphs of size $n$ and diameter $D$
is in $\Omega(n \log n)$.
\end{theorem}
\begin{proof}
If we consider a lollipop of diameter $D\le \alpha n$, then there are $\Omega(n^2)$ edges of the candy that are outside of the view at depth $D$ of the endpoint $u$ of the stick. 
The idea of the proof is based on the fact that  information about these edges has to be assigned to nodes of the graph as advice that will become available to $u$ within time $D$.

First observe that the view at depth $D$ of the endpoint $u$ of the stick is the same for all lollipops in $\cL(n,D)$.
Hence, if the size of advice is at most $c n \log n$, then the number of distinct decorated views of node $u$ is at most $2^{cn^2 \log n +n}<2^{cn^2 (\log n +1)}$.
We will show that, if {$c<(1-\alpha)^2/5$,}
then, for sufficiently large $n$, the number of lollipops in $\cL(n,D)$ exceeds this bound.
Indeed,  by Corollary~\ref{corLollipops}, the size of the class $\cL(n,D)$ is at least
{$$\frac{(\lceil n(1-\alpha) \rceil -1)!^{\lceil n(1-\alpha) \rceil}}{\lceil n(1-\alpha) \rceil !}$$}
For sufficiently large $n$ we have
\begin{small} 
{
$$ \frac{(\lceil n(1-\alpha) \rceil -1)!^{\lceil n(1-\alpha) \rceil}}{\lceil n(1-\alpha) \rceil !} > 
\frac{n(1-\alpha)}{2}^{n^2(1-\alpha)^2/4}\,.$$
}
\end{small}
It is enough to show that 
{$$ \left(\frac{n(1-\alpha)}{2}\right)^{n^2(1-\alpha)^2/4} > 2^{ n^2 (\log n +1)(1-\alpha)^2/5},$$}
which is immediate to verify by taking the logarithm of both sides.

It follows that the same decorated view at depth $D$ of node $u$ corresponds to different lollipops from $\cL(n,D)$. In view of Proposition~\ref{propSameHistory}, this proves that (even anonymous) topology recognition in time $D$, for the class of graphs of size $n$ and diameter $D\le \alpha n$, requires advice of size at least 
{$(n \log n) (1-\alpha)^2/5 \in \Omega(n \log n)$.}
\end{proof}

Proposition~\ref{ubTimeD} and Theorem \ref{lbD} imply the following corollary.
\begin{corollary}\label{corD}
Let $D\le \alpha n$.
The minimum size of advice sufficient to perform topology recognition in time $D$ in the class of graphs of size $n$ and diameter $D$
is in $\Theta(n \log n)$.
\end{corollary}

\section{Time below $D$}
In this section we study the minimum size of advice sufficient to perform topology recognition when the time allotted for this task is too short, for some node, to communicate with all other nodes in the network.


\algo{TR-4}
\noindent{\bf Advice:}\\
Let $G$ be a graph of size $n$ and diameter $D$.
The oracle assigns a unique label $\ell(v)$ from $\{0,\ldots, n-1\}$ to each node $v$ in the graph $G$.
It codes all edges of the graph as quadruples $\left<\ell(u),p,q,\ell(v)\right>$, where $\ell(u)$ and $\ell(v)$ are the labels of two adjacent nodes $u$ and $v$, $p$ is the port number at node $u$ corresponding to edge $\{u,v\}$, and $q$ is the port number at node $v$ corresponding to this edge.
Let $E$ be the set of all these codes.

Let $t = \lfloor (D-k)/3\rfloor ${, for $D\ge k >0$, be the time available to complete the topology recognition}.
If $t=0$, then the advice provided by  the oracle to each node $u$ is: $\ell(u)$, the collection $E$ of all edges, and the integer $0$.

If $t\ge 1$, then the oracle picks a set $X$ of nodes in  $G$ satisfying Lemma~\ref{lemDominating}, for $r=2t$.
{For each node $x\in X$, l}et $z(x) = |N_{t}(x)|$. Moreover, let $E_1, \ldots, E_{z(x)}$ be a partition of the edges in $E$ into $z(x)$ pairwise disjoint sets of sizes differing by at most 1. 
Let $v_1, \ldots,v_{z(x)}$ be an enumeration of nodes in $N_t(x)$.
The advice given by the oracle to node $v_i\in N_t(x)$ consists of the label $\ell(v_i)$, of the set $E_i$, and of the integer $t$.
Every other node $u$ only gets $\ell(u)$ and $t$ as advice.
Let $A$ be the resulting assignment of advice.

\noindent{\bf Node protocol:}\\
Let $t$ be the integer received by all nodes as part of their advice.

In round $i$, with $1\le i \le 3t$, each node sends to all its neighbors the collection of edges learned in all previous rounds. 
(In particular, if $t=0$, then there is no communication.)
After $3t$ rounds of communication each node reconstructs the topology and stops.
\procend

\begin{theorem}\label{ubTimeD-k}
Let $0< k \le D$.
Algorithm {\tt TR-4} completes topology recognition for all graphs of size $n$ and diameter $D$ within time $D-k$, using advice of size $O((n^2\log n)/(D-k+1))$.
\end{theorem}
\begin{proof}
If $t=0$ then $D-k$ is constant. In this case each node uses the collection $E$ of edges received as advice to reconstruct the topology in time $0$.
Since there are $O(n^2)$ edges in a graph of size $n$, and each edge can be coded with $O(\log n)$ bits, the size of advice is $O(n^2 \log n)=O((n^2 \log n)/(D-k+1))$.

Hence we can assume that $t\ge 1$.
Since each node is at distance at most $2t$ from some node in $X$, and the decorations of the nodes in $N_t(v)$, for each  such node $v$, assigned according to $A$, collectively contain the entire set $E$, after time $3t\le D-k$ each node of the graph can reconstruct the whole (labeled) topology.
It remains to be shown that the size of advice is in $O((n^2\log n)/(D-k+1))$. Indeed, $z(x) > t \in \Omega(D-k+1)$ and the size of each set $E_i$ is at most $\lceil |E| / z(x) \rceil$. Each edge can be coded with $O(\log n)$ bits and sets $N_t(x)$, for $x\in X$, are disjoint.
The theorem follows from the fact that $|E|\in O(n^2)$.
\end{proof}


The following lower bound shows that the size of advice used by Algorithm~{\tt TR-4} is asymptotically optimal.

\begin{theorem}\label{lbD-k}
Let $ D\le  \alpha n$ and $0 < k \le D$.
The size of advice needed to perform topology recognition in time $D-k$ in the class of graphs of size $n$ and diameter $D$
is in $\Omega((n^2 \log n)/(D-k+1))$.
\end{theorem}
\begin{proof}
The argument follows closely the proof of Theorem~\ref{lbD}. The only difference is that we consider decorated views of the endpoint $u$ of the stick at depth $D-k$ instead of $D$. Since there are only $D-k+1$ nodes in the lollipop at distance $D-k$ from $u$, if the  size of advice is at most  $(cn^2\log n)/(D-k+1)$, then the number of distinct decorated views of node $u$ at depth $D-k$ is at most  $2^{((cn^2 \log n)/(D-k +1)  + 1) (D-k +1)}= 2^{cn^2 \log n  + (D-k +1)} \le 2^{cn^2 (\log n  +1)}$, for sufficiently large $n$.

The same computation as in the proof of Theorem~\ref{lbD} shows that, if $c<(1-\alpha)^2/7$, then the same decorated view at depth $D-k$ of node $u$ corresponds to different lollipops from $\cL(n,D)$. In view of Proposition~\ref{propSameHistory}, this proves that (even anonymous) topology recognition in time $D-k$, for the class of graphs of diameter $D$ and size $n$, requires advice of size at least $((1- \alpha)^2/7) \cdot (n^2 \log n) /(D-k+1) \in \Omega((n^2 \log n)/(D-k+1))$.
\end{proof}

Theorems~\ref{ubTimeD-k} and \ref{lbD-k} imply the following corollary.
\begin{corollary}\label{corD-k}
Let $D\le \alpha n$  and $0 < k \le  D$.
The minimum size of advice sufficient to perform topology recognition in time $D-k$ in the class of graphs of size $n$ and diameter $D$
is in $\Theta((n^2\log n)/(D-k+1))$.
\end{corollary}

\section{Conclusion and open problems}
We presented upper and lower bounds on the minimum size of advice sufficient to perform topology recognition, in a given time $T$, in $n$-node networks of diameter $D$. 
Our bounds are asymptotically tight for time $T=2D+1$ and, if $D\le \alpha n$ for some constant $\alpha<1$, in the time interval $[0, \ldots , 3D/2]$.
Moreover, in the remaining time interval $(3D/2, \ldots, 2D]$ our bounds are still asymptotically tight if $D\in \Omega(\log n)$.
Closing the remaining gap between the lower bound $1$ and the upper bound $O(1+(\log n)/k)$ in this remaining time interval, for graphs of very small diameter $D \in o(\log n)$, is a natural open problem. In particular, it would be interesting to find the minimum time in which topology recognition can be accomplished using advice of constant size, or even of size exactly 1.

Other open problems remain in the case of networks with very large diameter, those which do not satisfy the assumption $D\le \alpha n$ for some constant $\alpha <1$, or equivalently those for which $n-D \in o(n)$.
Our upper bounds do not change in this case (we did not use the assumption $D\le \alpha n$ in their analysis), while our lower bounds change as follows, using the same constructions.
The lower bound for time above $D$, i.e., when $T=D+k$, where $0<k\le D$, becomes $\Omega((\log(n-D))/k)$;
our lower bound for time $D$ becomes $\Omega(((n-D)^2\log(n-D))/n)$;
the lower bound for time below $D$, i.e., when $T=D-k$, where $0<k\le D$, becomes $\Omega(((n-D)^2\log(n-D))/(D-k+1))$. It remains to close the gaps between these lower bounds and the upper bounds that we gave for each allotted time.

Another open problem is related to the communication model used. As mentioned in the Introduction, we chose the $\cal LOCAL$ model which ignores the size of messages. It would be interesting to 
study how the results change in the $\cal CONGEST$ model in which messages are limited to logarithmic size. Some compression of the information transmitted should be possible, e.g., along the lines of \cite{Ta}.

Let us also address the issue of node identities vs. advice given to nodes.
We did our study for unlabeled networks, arguing that nodes may be reluctant to disclose their identities for security or privacy reasons.
As we have seen, however, for anonymous networks some advice has to be given to nodes, regardless of the allotted time.
Does the oracle have to provide new distinct labels to nodes?
Our results show that for time above $D$ this is not the case, as the minimum size of advice enabling topology recognition in this time is too small for assigning a unique identifier to each node.
Hence, in spite of not having been given, a priori, unique identifiers, nodes can perform labeled topology recognition in this time span.
On the other hand for time at most $D$, the minimum size of advice is sufficiently large to provide distinct identifiers to nodes, and indeed our oracles inserted unique identifiers as part of advice. However, this should not raise concerns about security or privacy, as these identifiers may be arbitrary and hence should be considered as ``nicknames'' temporarily assigned to nodes. 

A different set of problems would appear if (unlike in our scenario) randomization were allowed. For many time values the size of required advice
would then be drastically reduced. For example, for time $D$, advice of size $O(\log n)$ would suffice. Indeed, providing nodes with
 the diameter $D$ and the size $n$ of the network would enable them to first acquire
different labels with high probability (the size $n$, or some bound on it, is needed to know how long the random labels should be, but such labels can be
randomly produced with no communication),
and then learn the topology of the resulting labeled graph in $D$ rounds. We leave the analysis of tradeoffs between time and advice for randomized topology 
recognition as an open problem.  
 
\bibliographystyle{plain}

\end{document}